\documentclass{article}
\usepackage{fullpage}
\usepackage{amsmath,amssymb,amsthm}
\usepackage{algorithm}
\usepackage[noend]{algorithmic}
\usepackage{hyperref}
\usepackage{cleveref}

\usepackage{bbm} %

\newcommand{\eps}{\varepsilon}
\newcommand{\R}{\mathbb{R}}
\DeclareMathOperator*{\E}{\mathbb{E}}

\usepackage{mathtools}
\DeclarePairedDelimiter{\norm}{\lVert}{\rVert}
\DeclarePairedDelimiter{\abs}{\lvert}{\rvert}

\DeclareMathOperator{\polylog}{polylog}

\DeclareMathOperator{\rank}{rank}
\DeclareMathOperator{\supp}{supp}
\DeclareMathOperator{\diag}{diag}
\DeclareMathOperator{\tr}{tr}
\newcommand{\one}{\ensuremath{\mathbbm{1}}}

\usepackage{thmtools}
\newtheorem{theorem}{Theorem}[section]
\newtheorem{lemma}[theorem]{Lemma}

\newtheorem{corollary}[theorem]{Corollary}

\newtheorem{proposition}[theorem]{Proposition}
\newtheorem{question}[theorem]{Question}

\newtheorem*{problem*}{Problem}
\newtheorem{remark}[theorem]{Remark}
\newtheorem*{remark*}{Remark}

\usepackage{xcolor}

\title{Online Algorithms for Spectral Hypergraph Sparsification}
\author{
Tasuku Soma \\ Institute of Statistical Mathematics \\ \texttt{soma@ism.ac.jp}
\and 
Kam Chuen Tung \\ University of Waterloo\\ \texttt{kctung@uwaterloo.ca}
\and 
Yuichi Yoshida \\ National Institute of Informatics \\ \texttt{yyoshida@nii.ac.jp} 
}
\begin{document}
\maketitle

\begin{abstract}
We provide the first online algorithm for spectral hypergraph sparsification. 
In the online setting, hyperedges with positive weights are arriving in a stream, and upon the arrival of each hyperedge, we must irrevocably decide whether or not to include it in the sparsifier. 
Our algorithm produces an $(\eps, \delta)$-spectral sparsifier with multiplicative error $\eps$ and additive error $\delta$ that has $O(\eps^{-2} n \log n \log r \log(1 + \eps W/\delta n))$ hyperedges with high probability, where $\eps, \delta \in (0,1)$, $n$ is the number of nodes, $r$ is the rank of the hypergraph, and $W$ is the sum of edge weights. The space complexity of our algorithm is $O(n^2)$, while previous algorithms require the space complexity of $\Omega(m)$, where $m$ is the number of hyperedges.
This provides an exponential improvement in the space complexity since $m$ can be exponential in $n$.
\end{abstract}

\section{Introduction}
\emph{Spectral sparsification} is a cornerstone of modern algorithm design.
The studies of spectral sparsification date back to the seminal work of Spielman~and~Teng~\cite{Spielman2011} for undirected graphs.
Let $G = (V, E, w)$ be an undirected graph with positive edge weight $w : E \to \R_{>0}$.
Let $\eps \in (0,1)$ be an arbitrary constant.
A weighted graph $\tilde G = (V, E, \tilde w)$ on the same node set $V$ is called an \emph{$\eps$-spectral sparsifier} of $G$ if
\[
    (1-\eps) z^\top L_G z \leq z^\top L_{\tilde G} z \leq (1+\eps) z^\top L_G z
\] 
for all $z \in \R^V$, where $\tilde w$ is a nonnegative edge weight, and $L_G$ and $L_{\tilde G}$ denote the Laplacian matrices of $G$ and $\tilde G$, respectively.
The number of nonzeros in $\tilde w$ is called the \emph{size} of a spectral sparsifier $\tilde G$.
Spielman~Teng~\cite{Spielman2011} showed that  one can find an $\eps$-spectral sparsifier with $O(\eps^{-2}n\log n)$ edges in nearly linear time in the size of the input graph.
Since then, there has been a series of works on spectral sparsification of graphs and various applications in the design of fast algorithms; see~\cite{Batson2013,Vishnoi2013} for survey.

Recently, the notion of spectral sparsification was extended to undirected hypergraphs and has been actively studied in the literature~\cite{Soma2019,Bansal2019,Kapralov2021,Kapralov2022,Lee2023,Jambulapati2023}.
For a weighted hypergraph $H = (V, E, w)$ with a positive edge weight $w : E \to \R_{>0}$, the energy function $Q_H: \R^V \rightarrow \R$ of $H$ is given by
\[
  Q_H(z) := \sum_{e \in E} w(e) \max_{u, v \in e} (z(u) - z(v))^2.
\]
This is a generalization of the quadratic form of the Laplacian matrix of a graph.
For $\eps \in (0,1)$, a weighted hypergraph $\tilde H = (V, E, \tilde w)$ on the same node set $V$ is called an $\eps$-spectral sparsifier of $H$ if 
\begin{equation} \label{eqn:eps-delta-sparsifier}
    (1-\eps)Q_{H}(z) \leq Q_{\tilde H}(z) \leq (1+\eps)Q_{H}(z)
\end{equation}
for all $z \in \R^V$.
Again, the number of nonzeros in $\tilde w$ is called the size of a spectral sparsifier $\tilde H$.
Since the number of hyperedges can be exponential, the existence of polynomial-size spectral sparsifiers is nontrivial.
This concept was first introduced by Soma and Yoshida~\cite{Soma2019}, and they showed that there exists an $\eps$-spectral sparsifier with $O(\eps^{-2}n^3)$ hyperedges and it can be found in time polynomial in the size of the input hypergraph.
The current best upper bound on the size of spectral sparsifiers is $O(\eps^{-2}n\log n \log r)$ by~\cite{Lee2023,Jambulapati2023}, where $r$ is the \emph{rank} of the hypergraph, i.e., the maximum size of a hyperedge.

However, all known algorithms for hypergraph spectral sparsification are \emph{offline}, i.e., they first store the entire hypergraph in the working memory and then construct a spectral sparsifier.
This is somewhat unreasonable because the space complexity (e.g., the size of input hypergraphs) could be exponentially larger than the size of the output sparsifier.
So we are naturally led to the following question:
Can we construct a spectral sparsifier of hypergraphs with smaller space complexity?

To formalize this, we study \emph{online} spectral sparsification in this paper.
In the online setting, the hyperedges $e_1, \dots, e_m$ arrive in a stream fashion together with their weights.
When edge $e_i$ arrives, we must decide immediately whether or not to include it in the sparsifier $\tilde{H}$.
Our goal is for $\tilde{H}$ to have a small number of edges and for the algorithm to use little working memory.

\subsection{Our Contribution}
We provide the first algorithm for online spectral hypergraph sparsification.
We say $\tilde H$ is an \emph{$(\eps, \delta)$-spectral sparsifier} of $H$, if
\[
  (1-\eps) Q_{H}(z) - \delta \norm{z}_2^2
  \le 
  Q_{\tilde{H}}(z)
  \le
  (1+\eps) Q_{H}(z) + \delta \norm{z}_2^2
\]
for all $z \in \R^V$.
Our main contribution is the following.

\begin{theorem}[Main]\label{thm:streaming-main}
    There exists an online algorithm (Algorithm~\ref{alg:streaming-i}) with the following performance guarantees:
  \begin{itemize}
      \item The amount of working memory required is $O(n^2)$ assuming the word RAM model;
      \item With high probability (i.e., probability at least $1 - 1/n$), it finds an $(\eps, \delta)$-spectral sparsifier $\tilde H$ of a rank-$r$ hypergraph $H$ with
      \[
        O \left( \frac{n \log n \log r}{\eps^2} \cdot
          \log \left(1 + \frac{\eps W}{\delta n} \right) \right)
      \]
      many hyperedges, where $W = \sum_{e \in E}w(e)$.
  \end{itemize}
\end{theorem}

\begin{remark}[Lower Bound]
  We remark that the upper bound on the number of hyperedges is tight up to logarithmic factors.
  In fact, in~\cite[Theorem 5.1]{Cohen2020} it is shown that even in the graph case it is necessary to sample $\Omega(n \log(1 + \eps W / (\delta n)) / \eps^2)$ edges.
\end{remark}

We can also obtain an $\eps$-spectral sparsifier for rank-$r$ hypergraphs if the range of edge weights is known in advance.
\begin{corollary}[$\eps$-spectral sparsifier]\label{cor:eps-sparsifier}
    Suppose that $H=(V,E,w)$ is a rank-$r$ hypergraph and that $W_{\min} \leq w(e) \leq W_{\max}$ for every $e \in E$ for some $0 < W_{\min} \leq W_{\max}$.
    Then, Algorithm~\ref{alg:streaming-i} with $\delta = O(\eps W_{\min}^2 n^{-2r})$ finds an $\eps$-spectral sparsifier with 
        \[ 
            O\left(\frac{n r \log n \log r}{\eps^2} \cdot \log\left(\frac{n W_{\max}}{W_{\min}}\right) \right)
        \] 
        hyperedges with high probability.
\end{corollary}

\subsection{Our Techniques}
We outline our algorithms below.
Our starting point is the work of spectral hypergraph sparsification via \emph{generic chaining}~\cite{Lee2023}.
He showed that if we sample each hyperedge with probability proportional to the effective resistance of an auxiliary (ordinary) graph, then the resulting hypergraph is a spectral sparsifier with high probability.
Here, the auxiliary graph is a weighted clique-graph $G = (V, F)$, where $F$ is  the multiset of undirected edges obtained by replacing every hyperedge $e \in E$ with the clique on $V(e)$.
The weight of edge $(u,v)$ coming from hyperedge $e$ is given by $w(e) c_{e,u,v}$, where $c: F \to \R_{\ge 0}$ is a special \emph{reweighting} satisfying the following conditions~\cite{Kapralov2022}:
For each $e \in E$, (i) $\sum_{u, v \in e} c_{e,u,v} = 1$ for $e \in E$ and (ii) $c_{e,u,v} > 0$ implies $r_{u,v} = \max_{u',v' \in e} r_{u',v'}$, where $r_{u,v}$ denotes the effective resistance between $u$ and $v$ in $G$.
Such a reweighting can be found by solving the following convex optimization problem
\begin{alignat*}{2}
    \text{maximize} & \quad \log\det\left( \sum_{e \in E}\sum_{u,v \in e} w(e) c_{e,u,v} L_{uv} + J \right) \\
    \text{subject to} & \quad \sum_{u,v \in e} c_{e,u,v} = 1 \quad (e \in E) \\
    & \quad c_{e,u,v} \geq 0 \quad (e\in E, u,v \in e),
\end{alignat*}
where $L_{uv}$ is the Laplacian of edge $(u,v)$ and $J$ is the all-one matrix~\cite{Lee2023}.
Now we set the edge sampling probability $p_e \propto w(e) \max_{u,v \in e} r_{u,v}$.
Using the generic chaining technique,~\cite{Lee2023} showed that this yields an $\eps$-spectral sparsifier with a constant probability.

In the online setting, the entire hypergraph is not available, so we have to estimate the edge sampling probability on the fly.
Inspired by the online row sampling algorithm~\cite{Cohen2020}, we introduce the \emph{ridged edge sampling probability}.
Let $\eta = \eps/\delta$.
For $i = 1, \dots, m$, we iteratively compute a sequence of auxiliary graphs $G_i$.
Initially, $G_0$ is the empty graph on $V$.
For each $i > 0$, we construct a graph $G_i$ from $G_{i-1}$ by adding an edge $(u,v)$ with weight $w_i c_{i,u,v}$ for every pair of vertices $(u,v)$ in $e_i$, where $w_i$ is the weight of $e_i$ and $c_{i,u,v}$ is an optimal solution of the following convex optimization problem:
\begin{alignat*}{2}
    \text{maximize} & \quad \log\det\left( L_{G_{i-1}} + \sum_{u,v \in e_i} w_i c_{i,u,v} L_{uv} + \eta I \right) \\
    \text{subject to} & \quad \sum_{u,v \in e_i} c_{i,u,v} = 1 \\
    & \quad c_{i,u,v} \geq 0 \quad (u,v \in e_i)
\end{alignat*}
Then, we define the $\eta$-ridged edge sampling probability by $p_i \propto \max_{u,v \in e_i} \norm{(L_{G_i} + \eta I)^{-1/2}(\chi_u - \chi_v)}_2^2$, where $\chi_u$ denotes the $u$-th standard unit vector.
Using the techniques from~\cite{Jambulapati2023b}, we show that this gives an $(\eps, \delta)$-spectral sparsifier having the desired number of hyperedges with high probability.
Since we only need to maintain the Laplacian of $G_i$, the space complexity is $O(n^2)$ as required.
Note that the above convex optimization problem can be solved by projected gradient descent (up to desired accuracy), which only requires space of linear in the dimension, i.e., $O(n^2)$.

\subsection{Related Work}
The literature on spectral sparsification is vast.
We refer the readers to~\cite{Batson2013,Vishnoi2013} for technical details and various applications.
Spectral sparsification of hypergraphs can be used to speed up semi-supervised learning with hypergraph regularizers and hypergraph network analysis; see discussion in~\cite{Soma2019,Kapralov2022} for further applications.

Spectral sparsification for graphs in the \emph{semi-streaming} setting is also well-studied.
This is almost identical to our online setting, but the algorithm can change the weights of edges in the output that have already been sampled.
Kelner~and~Levin~\cite{Kelner2013} initiated this line of research and provided a natural extension of the celebrated effective resistance sampling sparsification algorithm of~\cite{Spielman2011b}.\footnote{As pointed out in~\cite{Cohen2020}, the original analysis has a subtle dependency issue. Later,~\cite{Calandriello2016} provided a complete analysis of their algorithm.}
Cohen~et~al.~\cite{Cohen2020} devised an online row sampling algorithm for general tall and skinny matrices in the online setting, which includes online spectral sparsification of graphs.
Kapralov~et~al.~\cite{Kapralov2017} devised a fully dynamic streaming algorithm for spectral sparsification of graphs, that supports both insertion and deletion.

Beyond undirected hypergraphs, there are several works on spectral sparsification for more complex objects such as directed hypergraphs~\cite{Soma2019,Oko2023}, submodular functions~\cite{Rafiey2022}, and the sum of norms~\cite{Jambulapati2023b}.

\section{Preliminaries}

\subsection{Notations}

We use $\R_{\ge 0}$ and $\R_{> 0}$ to denote the set of nonnegative and positive real numbers, respectively.
Given a positive integer $N$, we use $[N]$ to denote the set $\{1, 2, \dots, N\}$.
All logarithms are natural logarithms unless otherwise specified.
Given a function $f: X \rightarrow \R$, its support $\supp(f)$ is the set $\{x \in X: f(x) \neq 0\}$.
For any (finite) set $X$ and element $u \in X$, the vector $\chi_u \in \R^X$ is the vector whose $u$-entry is $1$ and all other entries are $0$, and the vector $\one_X$ (or simply $\one$) is the vector whose entries are all $1$.
For $x, x' \in \R^X$, $x \perp x'$ denotes $\langle x, x' \rangle = 0$.
For $0 \le p \le 1$, $\mathrm{Ber}(p)$ denotes the random variable that takes value $1$ with probability $p$ and $0$ with probability $1-p$.

\subsection{Linear Algebra}
Let $M \in \R^{k \times k}$ be a symmetric matrix. By the spectral theorem, $M$ has the following decomposition $M = V \Lambda V^\top$, where $V$ is orthogonal and $\Lambda = \diag(\lambda_1, \dots, \lambda_k)$ is the diagonal matrix consisting of the eigenvalues of $M$. It is known that $\Lambda$ is unique up to permutation.
We say that $M$ is \emph{positive definite} (denoted $M \succ 0$) if $\lambda_i > 0$ for all $i$, and
we say that $M$ is \emph{positive semidefinite (PSD)} (denoted $M \succeq 0$) if $\lambda_i \ge 0$ for all $i$.

Given a PSD matrix $M = V \Lambda V^\top$, its pseudoinverse $M^{\dag}$ is $V \Lambda^{\dag} V^\top$ where $\Lambda^{\dag}_{i, i} = \lambda_i^{-1}$ if $\lambda_i > 0$ and $\Lambda^{\dag}_{i, i} = 0$ if $\lambda_i \ge 0$.
Accordingly, $M^{\dag/2} = \sqrt{M^{\dag}} = (\sqrt{M})^{\dag}$ where $\sqrt{\cdot}$ is the usual matrix square root defined on PSD matrices.

\subsection{Hypergraphs}

\textbf{Basic definitions.} A hypergraph $H = (V, E, w)$ is defined with a vertex set $V$, 
a (multi-)set of hyperedges $E = \{e_1, e_2, \dots, e_{|E|}\}$, where each $e_i \subseteq V$ and $|e_i| \ge 2$,
and a weight function $w: e \in E \mapsto w(e) \in \R_{\ge 0}$. We use $w_i$ as a shorthand for $w(e_i)$. 
Unless otherwise specified, $n := |V|$ denotes the number of vertices and $m := |E|$ denotes the number of hyperedges of $H$.
Denote by $r = \rank(H) := \max_{i \in [m]} |e_i|$ the \emph{rank} of $H$, i.e., the size of the largest hyperedge in $H$.
We say that $H$ is \emph{unweighted} if $w_i = 1$ for all $i \in [m]$.\\

\noindent
\textbf{Energy and sparsification.} Given a hyperedge $e \subseteq V$, its energy $Q_e: \R^V \rightarrow \R$ is defined as:
\[
  Q_e(z) := \left[ \max_{u, v \in e} (z(u) - z(v)) \right]^2 =  \max_{u, v \in e} (z(u) - z(v))^2.
\]
The energy $Q_H: \R^V \rightarrow \R$ of a hypergraph $H$ is the weighted sum of its edge energies:
\[
  Q_H(z) := \sum_{i \in [m]} w_i Q_{e_i}(z) = \sum_{i \in [m]} w_i \max_{u, v \in e_i} (z(u) - z(v))^2.
\]

\noindent
Given a hypergraph $H = (V, E, w)$ and error parameters $\eps, \delta \ge 0$ ($\eps$ for relative error, $\delta$ for absolute error), an $(\eps, \delta)$-spectral sparsifier of $H$ is a hypergraph $\tilde{H} = (V, E, \tilde{w})$ that is supported on the hyperedge set $E$ and its energy $Q_{\tilde{H}}$ satisfies
\[
  (1 - \eps) Q_{H}(z) - \delta z^\top z \le Q_{\tilde{H}}(z) \le (1 + \eps) Q_{H}(z) + \delta z^\top z
  \qquad \forall z \in \R^V.
\]
The size of the sparsifier $\tilde{H}$ is simply $|\supp(\tilde{w})|$.
The term ``$\eps$-spectral sparsifier'' refers to $(\eps, 0)$-spectral sparsifiers, i.e., additive error is not allowed.
It corresponds to the well-established notion of spectral sparsifiers introduced in~\cite{Spielman2011}.

\subsection{Reweighting}

A reweighting of a hyperedge $e \subseteq V$ is a set of weights $\{c_{u, v}\}_{u, v \in e}$ such that $c_{u, v} \ge 0$ and $\sum_{u, v \in e} c_{u, v} = 1$.
The corresponding reweighted clique-graph $G$ is the (ordinary) graph on $V$ with edges $(u, v)$ having weight $c_{u, v}$ (and other edges having zero weight).
A reweighting of a hypergraph $H$ is the weighted sum of the reweightings of its hyperedges. The corresponding reweighted clique-graph $G$ is the graph on $V$ with edges $(u, v)$ having weight $w(u, v) = \sum_{i \in [m]: u, v \in e_i} w_i c_{i, u, v}$, where $\{c_{i, u, v}\}_{i \in [m]: u, v \in e_i}$ is a reweighting of the edge $e_i$.
\\
The reason for considering reweighted clique-graphs is that the Laplacian
\[
  [L_G(z)]_u := \sum_{v \in V} w(u, v) (z(u) - z(v))
\]
of an ordinary graph is a linear operator. More importantly, it is PSD.
Its energy $Q_G$ is then
\[
  Q_G(z) := \langle z, L_G(z) \rangle
  =\sum_{u, v \in V} \sum_{i \in [m], u, v \in e_i} w_i c_{i, u, v} (z(u) - z(v))^2
  = \sum_{i \in [m]} w_i \sum_{u, v \in e_i} c_{i, u, v} (z(u) - z(v))^2.
\]

The following property follows from the fact that the energy of a reweighted clique-graph of $H$ is at most that of $H$. 
\begin{proposition}[Energy Comparison]
  \label{prop:energy-comparison}
  Let $H$ be a hypergraph and $G$ be a reweighted clique-graph of $H$. Let $x \in \R^V$ such that $x \perp \one$. Then,
  $Q_H(L_G^{\dagger/2}x) \ge \norm{x}_2^2$.
\end{proposition}
\begin{proof}
By the definition of a reweighting,
\begin{align*}
    Q_H(x) 
    &= \sum_{i=1}^m w_i \max_{u,v \in e}{(x(u) - x(v))^2} \\
    &= \sum_{i=1}^m \sum_{u,v \in e} w_i c_{i,u,v} \max_{u,v \in e}{(x(u) - x(v))^2} \\
    &\geq \sum_{i=1}^m \sum_{u,v \in e} w_i c_{i,u,v} (x(u) - x(v))^2 \\
    &= x^\top L_G x.
\end{align*}
So, $Q_H(L_G^{\dag/2}x) \geq \norm{x}_2^2$.
\end{proof}

\subsection{Generic Chaining}

Our sparsifier $\tilde{H}$ will be an unbiased random sample of $H$.
Since the hypergraph energies $Q_H$ and $Q_{\tilde{H}}$ in \eqref{eqn:eps-delta-sparsifier} are not quadratic forms, i.e. not of the form $z \mapsto z^\top M z$ where $M$ is a matrix, we cannot use matrix concentration inequalities to control the deviation $|Q_{\tilde{H}}(z) - Q_{H}(z)| / Q_H(z)$ at all points $z$. Instead, we prove pointwise concentration and extend it to a uniform bound over the entire domain using Talagrand's generic chaining \cite{Talagrand2014}.
We summarize here certain useful facts about generic chaining.
Let $(V_x)_{x \in X}$ be a real-valued stochastic process and $d$ be a semi-metric on $X$.
We say $(V_x)_{x \in X}$ is a \emph{subgaussian process with respect to $d$} if
\[
    \Pr(V_x - V_y > \eps) \leq \exp\left(-\frac{\eps^2}{2d(x, y)^2}\right)   
\]
for all $x, y \in X$ and $\eps > 0$.
Talagrand's generic chaining relates the supremum of the process $(V_x)$ with the following geometric quantity, called the $\gamma$-functionals:
\[
  \gamma_2(X, d) := \inf_{\mathcal{X} = (X_h)} \sup_{x \in X} \sum_{h \ge 0} 2^{h/2} d(x, X_h),
\]
where the infimum is taken over all collections $(X_h)$ of \textit{admissible sequences}, meaning that $X_h \subseteq X$ and $|X_h| \le 2^{2^h}$ for all $h \ge 0$.
Intuitively, $\mathcal{X}$ is a successively finer net over which the union bound is applied.

The following lemma is key to obtaining a tail bound on the supremum of the process $(V_x)$ and will yield high-probability success guarantees for our streaming algorithm.
\begin{lemma}[{\cite[Theorem 2.2.27]{Talagrand2014}}]
   \label{lem:chaining-whp}
    Let $(V_x)_{x \in X}$ be a subgaussian process on a semi-metric space with respect to a semi-metric $d$
    and let $\Delta(X, d) := \sup_{x, y \in X} d(x, y)$ be the $d$-diameter of $X$.
    If $Z = \sup_{x \in X} |V_x|$, then for any $\lambda > 0$,
    \[
      \log \E[e^{\lambda Z}] \lesssim \lambda^2 \Delta(X, d)^2 + \lambda \gamma_2(X, d).
    \]
\end{lemma}

\subsection{Concentration Inequalities}

We shall use Azuma's inequality to establish the subgaussian bound required for our chaining argument.
\begin{proposition}[Azuma's Inequality]\label{prop:azuma}
  Let $\psi_0 = 0, \psi_1, \dots, \psi_T$ be a martingale. Suppose that
  \[
    |\psi_i - \psi_{i-1}| \le M_i
    \qquad \forall i \in [T].
  \]
  Then,
  \[
    \Pr \left[ |\psi_T| \ge \beta \right] \le 2 \exp \left( \frac{- \beta^2}{2 \sum_i M_i^2} \right).
  \]
\end{proposition}
The proof and the requisite background can be found in standard treatises such as \cite{CL2006}.

The following special case of Chernoff bound is also useful.
\begin{proposition}[Chernoff bound]
  \label{prop:chernoff}
  Let $X_1, \dots, X_T$ be independent Bernoulli random variables where $X_i = \mathrm{Ber}(p_i)$. Let $\mu := \sum_i p_i$. Then, for any $M > 0$,
  \[
   \Pr\left[\sum_i X_i > (\mu + M)\right] \le
   \inf_{\lambda > 0} \exp\left( \mu (e^\lambda - 1 - \lambda) - M \lambda \right).
  \]
\end{proposition}

\section{Algorithm Description}

In order to obtain an \textit{unbiased} estimator $\tilde{H}$ in this setting, our options are limited. We consider the class of algorithms where the current edge $e_i$ is sampled with probability $p_i$, where $p_i$ depends on all edge arrivals and decisions so far. If sampled, the edge is added to the sparsifier $\tilde{H}$ with weight $w_i / p_i$.\\

Our proposed algorithm, Algorithm~\ref{alg:streaming-i}, has the following features:
\begin{itemize}
    \item The sampling probability \textbf{does not} depend on previous decisions, but only on previous edge arrivals.
    \item It requires maintaining a reweighted graph $G_i$ of the hypergraph $H_i$ at all times, and uses $G_i$ to define sampling probabilities.
\end{itemize}

Using the reweighted graph $G_i$, the algorithm produces \textit{overestimates} of the ``importance'' of the hyperedge $e_i$ in the entire hypergraph $H$, using the effective resistances of the clique edges in the reweighted graph.
By virtue of it being an overestimate, it is relatively easy to analyze the success probability.
The difficulty lies in choosing an appropriate reweighting, so that the number of selected hyperedges remains well-controlled. This is why it is helpful to use the log-determinant potential function to guide the search for a suitable reweighting.

\begin{algorithm}[ht]\caption{Online Hypergraph Sparsification}\label{alg:streaming-i}
    \textbf{Input:}
    Hypergraph $H = (V, E)$ given as a stream, $\eps, \delta > 0$. Let $\eta := \delta / \eps$. \\
    \textbf{Initialization:} 
    let $G_0$ be the empty graph on $V$, $L_0^{\eta} := \eta I_n + L_{G_0} = \eta I_n$.\\
    \textbf{For} $i=1, \dots, m$
    \begin{enumerate}
        \item Edge $e_i$ arrives with weight $w_i$.
        \item Compute a reweighting $c_{i, u, v}$ of edge $e_i$, so that
        \[
          \log \det \left( L_{i-1}^{\eta} + \sum_{u, v \in e_i} w_i c_{i, u, v} L_{uv} \right)
        \]
        is maximized.
        \item Let $G_i$ be the graph obtained from $G_{i-1} $ by adding an edge $(u,v)$ with weight $ w_i c_{i, u, v}$ for every pair of vertices $(u,v)$ in $e_i$, and let $L_i^{\eta} := \eta I_n + L_{G_i} = L_{i-1}^{\eta} + \sum_{u, v \in e_i} w_i c_{i, u, v} L_{uv}$.
        \item Let $r_i := \max_{u, v \in e_i} \norm{(L_i^{\eta})^{-1/2} (\chi_u - \chi_v)}_2^2$ be the maximum ridged effective resistance across a pair of vertices in $e_i$.
        \item Sample edge $e_i$ with probability $p_i := \min(1, c r_i w_i)$, where $c = O(\eps^{-2}\log n \log r )$.
    \end{enumerate}
\end{algorithm}

\section{Analyzing the Success Probability}\label{sec:analysis-p-success}
In this section, we prove that Algorithm~\ref{alg:streaming-i} succeeds with high probability.
Given a hypergraph $H$, let
\[
  Q_H^{\eta}(z) := Q_H(z) + \eta z^\top z 
\]
be the $\eta$-ridged energy of $z$.
We would like to control the probability that $Q_{\tilde{H}}^{\eta}(z)$ is within a multiplicative factor of $1 \pm \eps$ from $Q_{H}^{\eta}(z)$ for all $z \in \R^V$.
Note that the event
\[
  \sup_{z: Q_H^{\eta}(z) \le 1} |Q_{\tilde{H}}^{\eta}(z) - Q_H^{\eta}(z)| \le \eps
\]
is the same as
\[
  (1 - \eps) Q_H^{\eta}(z) \le Q_{\tilde{H}}^{\eta}(z) 
  \le (1 + \eps) Q_H^{\eta}(z)
  \qquad \forall z \in \R^V,
\]
which by the choice of $\eta$ implies that
\[
  (1-\eps) Q_H(z) - \delta z^\top z
  \le Q_{\tilde{H}}(z)
  \le (1 + \eps) Q_H(z) + \delta z^\top z
  \qquad \forall z \in \R^V,
\]
i.e., that $\tilde{H}$ is an $(\eps, \delta)$-spectral sparsifier of $H$.\\

Our plan is as follows.
For the desired concentration bound, we will bound the exponential moment generating function (MGF)
$\E_{\tilde{H}} [\exp(\lambda \cdot \sup_{z: Q_H^{\eta}(z) \le 1} |Q_{\tilde{H}}^{\eta}(z) - Q_H^{\eta}(z)|)]$ of the energy discrepancy function.
By Markov's inequality and a suitable choice of $\lambda$, we can then conclude that $Q_H^{\eta}$ and $Q_{\tilde{H}}^{\eta}$ are $\eps$-close with high probability.

The following is the main technical result of the section:
\begin{lemma}[Exponential MGF Bound]\label{lem:exponential-mgf-bound}
    Let $H$ be a hypergraph stream and $\tilde{H}$ be the sampled hypergraph obtained from Algorithm~\ref{alg:streaming-i}. Let $Z := \sup_{z: Q_H^{\eta}(z) \le 1} |Q_{\tilde{H}}^{\eta}(z) - Q_H^{\eta}(z)|$.
    Then, for any $\lambda > 0$,
    \[
      \E_{\tilde{H}} \left[
      \exp (\lambda Z) \right]
      \le
      \E_{\tilde{H}} \left[
        \exp\left(
          \frac{\lambda^2}{c} (1 + Z)
          +
          \lambda \sqrt{\frac{\log n \log r}{c}} (1 + Z)^{1/2}
        \right)
      \right].
    \]
\end{lemma}

This is an \emph{implicit} bound on the exponential MGF of $Z := \sup_{z: Q_H^{\eta}(z) \le 1} |Q_{\tilde{H}}^{\eta}(z) - Q_H^{\eta}(z)|$ because $Z$ itself appears in the bound. 

The proof of \autoref{lem:exponential-mgf-bound} follows closely the chaining proofs in~\cite{Lee2023} and~\cite{Jambulapati2023b}, with a few modifications.
For completeness, we present the proof in full detail in \autoref{app:main-chaining-proof}.
In the remainder of the section, we will show how \autoref{lem:exponential-mgf-bound} implies that $Z \le \eps$ with high probability.

\subsection{High Probability Guarantee from \texorpdfstring{\autoref{lem:exponential-mgf-bound}}{Exponential MGF Bound}}
In order to obtain a high probability guarantee on the success probability, we must derive an \emph{explicit} upper bound on $\E_{\tilde{H}} \exp(\lambda Z)$.

Suppose $c = (\kappa_1 \log n \log r) / \eps^2$ and take $\lambda := (\kappa_2 \sqrt{\log n \log r}) / \eps$, where $\kappa_1$ is an absolute constant and $\kappa_2$ depends only on $\eps$. Let $\kappa := \kappa_2 / \sqrt{\kappa_1}$.
We will ensure that the parameters satisfy
$\kappa^2 + \kappa \sqrt{\log n \log r} \le \lambda / 2$. Then,
\begin{align*}
  \E_{\tilde{H}} \left[
      \exp (\lambda Z) \right]
  &\le
  \E_{\tilde{H}} \left[
    \exp\left(
      \frac{\lambda^2}{c} (1 + Z)
      +
      \lambda \sqrt{\frac{\log n \log r}{c}} (1 + Z)^{1/2}
    \right)
  \right]
  \tag{\autoref{lem:exponential-mgf-bound}}
  \\
  &=
  \E_{\tilde{H}} \left[
    \exp\left(
      \kappa^2 (1 + Z)
      +
      \kappa \sqrt{\log n \log r} (1 + Z)^{1/2}
    \right)
  \right]
  \\
  &\le
  \E_{\tilde{H}} \left[
    \exp\left(
      \kappa^2 (1 + Z)
      +
      \kappa \sqrt{\log n \log r} (1 + Z)
    \right)
  \right]
  \\
  &=
  \exp\left(
    \kappa^2 + \kappa \sqrt{\log n \log r}
  \right)
  \cdot
  \E_{\tilde{H}} \left[
    \exp\left(
      (\kappa^2 + \kappa \sqrt{\log n \log r}) Z
    \right)
  \right]
  \\
  &\le
  \exp\left(
    \kappa^2 + \kappa \log n
  \right)
  \cdot
  \E_{\tilde{H}} \left[
    \exp\left(\lambda Z
    \right)
  \right]^{\frac{\kappa^2 + \kappa \sqrt{\log n \log r}}{\lambda}}
  \tag{Jensen's inequality and $\frac{\kappa^2 + \kappa \sqrt{\log n \log r}}{\lambda} \le 1/2$}
  \\
  &\le
  \exp\left(
    \kappa^2 + \kappa \log n
  \right)
  \cdot
  \E_{\tilde{H}} \left[
    \exp\left(\lambda Z
    \right)
  \right]^{1/2},
\end{align*}
which resolves to
\[
  \E_{\tilde{H}}
  \left[
      \exp (\lambda Z)
  \right]
  \le
  \exp(2 \kappa^2 + 2 \kappa \log n).
\]
Markov's inequality then implies that
\[
 \Pr_{\tilde{H}} [Z > \eps]
 \le \frac{\E_{\tilde{H}}[\exp(\lambda Z)]}{e^{\lambda \eps}}
 \le \exp \left( 2 \kappa^2 + 2 \kappa \log n - \lambda \eps \right) = \exp \left( 2 \kappa^2 - (\kappa_2 - 2 \kappa) \log n \right).
\]
We can take $\kappa_1 \ge 16$ and $\kappa_2 \lesssim \min(\sqrt{\kappa_1}, 1/\eps)$  for the above analysis to go through, and for large enough $n$ or small enough $\eps$ (or both) satisfying $(2 \kappa^2 - (\kappa_2 - 2 \kappa) \log n) \le - \log n$. 
We arrive at the following conclusion.
\begin{lemma}[Success probability]\label{lem:sparsifier-success-probability}
 For Algorithm~\ref{alg:streaming-i}, let $Z := \sup_{z: Q_H^{\eta}(z) \le 1} |Q_{\tilde{H}}^{\eta}(z) - Q_H^{\eta}(z)|$. Then, we have $\Pr_{\tilde{H}} [Z > \eps] \le 1/n$.
  As a corollary, with probability $\ge 1-1/n$, $\tilde{H}$ is an $(\eps, \delta)$-spectral sparsifier of $H$.
\end{lemma}

\section{Bounding the Sample Size}

The expected number of edges in $\tilde{H}$ is simply $\sum_i p_i \le \sum_i c w_i r_i$, but it is not easy to bound each $r_i$ directly.
To bound this more easily, we use a potential function
\[
  \Phi_i := \log \det \left( L_{i}^{\eta} \right).
\]

We will show that, every time an edge gets sampled, $\Phi_i$ increases substantially.
Then, we bound the value of $\Phi_m - \Phi_0$, which will in turn give a bound on the number of edges.

Since the update to $L_{G_i}$ is no longer rank-1, but rank-$|e_i|$, we will make use of concavity of the log-determinant function:

\begin{proposition}\label{prop:log-det-concave}
    The function $X \mapsto \log \det(X)$ is concave on the set of positive definite matrices.
\end{proposition}

\begin{proof}
    Let's just check concavity over all rays $X + \lambda Y$, where $X$ is positive definite, $Y$ is a symmetric matrix and $\lambda$ is in a small enough open interval containing $0$. We have
    \begin{eqnarray*}
      \log \det(X + \lambda Y) - \log \det (X)
      &=&
      \log \det(I + \lambda X^{-1/2} Y X^{-1/2})
      \\
      &=&
      \sum_{i=1}^n \log (1 + \lambda c_i),
    \end{eqnarray*}
    where $c_i \in \R$ are the eigenvalues of $X^{-1/2} Y X^{-1/2}$. Note that $\log \det(X)$ is constant;
    the result then follows from the concavity of the function $\lambda \mapsto \log(1 + \lambda c_i)$ near $\lambda = 0$ for each $i$.
\end{proof}

\begin{proposition}[Potential Increase]\label{prop:potential-increase}
  We have
  \[
    \Phi_i - \Phi_{i-1} \ge \frac{p_i \log 2}{c}.
  \]
\end{proposition}

\begin{proof}
    Let us write $L_i^{\eta} := L_{i-1}^{\eta} + w_i \sum_{u, v \in e_i} c_{i, u, v} L_{uv}$.
    Let $R_i(u,v)$ denote the maximum ridged effective resistance between $u$ and $v$.
    By a KKT-condition argument similar to~\cite[Section 3.3]{Lee2023}, an optimal solution $c_{i, u, v}$ satisfies that
    $c_{i, u, v} > 0$ implies $R_i(u, v) = r_i := \max_{u', v' \in e_i} R_i(u', v')$. 
    Then,
    \begin{eqnarray*}
      \Phi_i - \Phi_{i-1}
      &=&
      \log \det(L_{i-1}^{\eta} + w_i \sum_{u, v \in e_i} c_{i, u, v} L_{uv}) - \log \det(L_{i-1}^{\eta})
      \\
      &\ge&
      \sum_{u, v \in e_i} c_{i, u, v} \left(
        \log \det(L_{i-1}^{\eta} + w_i L_{uv}) - \log \det(L_{i-1}^{\eta})
      \right)
      \qquad (\text{by \autoref{prop:log-det-concave}})
      \\
      &=&
      \sum_{u, v \in e_i} c_{i, u, v} \log \det \left( I + w_i (L_{i-1}^{\eta})^{-1} L_{uv} \right)
      \\
      &=&
      \sum_{u, v \in e_i} c_{i, u, v} \log(1 + w_i R_{i-1}(u, v))
      \\
      &\ge&
      \sum_{u, v \in e_i} c_{i, u, v} \log(1 + w_i R_i(u, v))
      \qquad (\text{Eff.\ resistance only decreases with edge addition})
      \\
      &=&
      \sum_{u, v \in e_i} c_{i, u, v} \log(1 + w_i r_i)
      \\
      &\ge&
      \log \left(1 + \frac{p_i}{c} \right)
      \\
      &\ge&
      \frac{p_i \log 2}{c}.
    \end{eqnarray*}
    The final inequality uses the fact that $\log(1 + x) \ge x \log 2$ for $x \in [0, 1]$. We know $p_i \le 1$ and we can easily set $c \ge 1$.
    This concludes the proof.
\end{proof}

\begin{lemma}[Number of sampled edges]\label{lem:sparsifier-size-bound}
We have
    \[
     \E[ |\tilde{H}| ] \lesssim c n \log \left( 1 + \frac{2W}{\eta n} \right),
    \]
where $W = \sum_{i=1}^m w_i$. Moreover,
\[
  |\tilde{H}| \lesssim c n \log \left( 1 + \frac{2W}{\eta n} \right)
\]
with probability at least $1 - 1/n$.
\end{lemma}

\begin{proof}
    We first bound the expected number of sampled edges. It follows rather straight-forwardly from \autoref{prop:potential-increase}.
    \begin{align*}
      \E[|\tilde{H}|]
      = \sum_{i \in [m]} p_i
      &\lesssim
      c \cdot (\Phi_m - \Phi_0)
      \\
      &=
      c \cdot \left(
       \log \det (L_m + \eta I) - \log \det(\eta I)
      \right)
      \\
      &= c \cdot \log \det (I + \eta^{-1} L_m).
    \end{align*}
    By the AM-GM inequality, we have
    \[
        \det(I + \eta^{-1}L_m) = \prod_{i=1}^n (1 + \eta^{-1}\lambda_i(L_m))
      \leq \left( 1 + \frac{1}{\eta n}\sum_{i=1}^n \lambda_i(L_m) \right)^n
      = \left( 1 + \frac{\tr L_m}{\eta n} \right)^n.
    \]
    Plugging this into the previous inequality, we get
    \begin{align*}
      \E[|\tilde{H}|]
      &\lesssim
      cn \cdot \log \left( 1 + \frac{\tr L_m}{\eta n} \right).
    \end{align*}
    Finally, 
    \[
        \tr L_m = \sum_{i=1}^m w_i \sum_{u, v \in e_i} c_{i, u, v} \tr{((\chi_u - \chi_v)(\chi_u - \chi_v)^\top)}
        = 2 \sum_{i=1}^m w_i \sum_{u, v \in e_i} c_{i, u, v} = 2W,
    \]
    and the desired bound on $\E|\tilde{H}|$ is established.
    The high probability guarantee then follows from a standard application of \autoref{prop:chernoff}.
\end{proof}

Combining \autoref{lem:sparsifier-success-probability} and \autoref{lem:sparsifier-size-bound} yields \autoref{thm:streaming-main}.

\paragraph{Proof of \autoref{cor:eps-sparsifier}.}
Now we prove \autoref{cor:eps-sparsifier}.
By the hypergraph Cheeger inequality~\cite{Chan2018}, $Q_H(x) \gtrsim W_{\min}^2n^{-2r} \norm{x}_2^2$ for $x \in \R^n$ with $x \perp \one$. 
Since $\delta = O(\eps W_{\min}^2 n^{-2r})$, we have $\eps Q_H(x) \geq \delta\norm{x}_2^2$.
So an $(\eps, \delta)$-spectral sparsifier is indeed a $(2\eps, 0)$-spectral sparsifier.
Therefore, Algorithm~\ref{alg:streaming-i} outputs an $(2\eps, 0)$-spectral sparsifier with a constant probability.
The expected number of hyperedges is immediate from \autoref{thm:streaming-main}.

\section{Conclusion}
To summarize, we designed and analyzed the first online algorithm for hypergraph spectral sparsification, showing that it uses significantly less space than the number of edges.
We leave open the following questions concerning the performance of the algorithm:

\begin{question}
  Can we derive a matching lower bound on the space complexity of any (online) streaming algorithm for spectral hypergraph sparsification?
\end{question}

\begin{question}
  Can we improve the space complexity from $O(n^2)$ to $O(n r \polylog m)$, or even better?
  Such an algorithm would perform better when the rank of the hypergraph is small.
\end{question}

While this paper focused on the insertion-only setting, the fully dynamic setting is also of interest.

\begin{question}
  Can we obtain an efficient fully dynamic algorithm (i.e., one that supports both hyperedge insertion and deletion) for spectral hypergraph sparsification?
\end{question}

\subsubsection*{Acknowledgements}
TS is supported by JSPS KAKENHI Grant Number JP19K20212.
A part of this work was done during KT's visit to National Institute of Informatics.
YY is supported by JSPS KAKENHI Grant Number JP20H05965 and JP22H05001.

\bibliographystyle{alpha}
\bibliography{main}

\appendix

\section{Proof of \autoref{lem:exponential-mgf-bound}}\label{app:main-chaining-proof}

\subsection{Symmetrization}
Chaining works best when the random variables involved are symmetric.
Therefore, we first make the quantity to be controlled more symmetric.
Since $Q_H^{\eta}(z) = \E_{\hat{H}} Q_{\hat{H}}^{\eta}(z)$, by Jensen's inequality we have
\[
\E_{\tilde{H}} \left[ \exp(\lambda Z)
  \right]
\le
\E_{\tilde{H}, \hat{H}} \left[
  \exp \left(
  \lambda \sup_{ Q_H^{\eta}(z) \le 1 }
    | Q_{\tilde{H}}^{\eta}(z) - Q_{\hat{H}}^{\eta}(z) |
  \right)
  \right].
\]
Here $\hat{H}$ is an independent copy of $\tilde H$.
We can write the inside term as
\[
  \sum_{i=1}^\top w_i \left( \frac{\tilde{\xi}_i}{p_i} - \frac{\hat{\xi}_i}{p_i} \right)
  \cdot Q_{e_i}(z)
\]
since the $\eta z^\top z$ terms cancel out one another.
Here, $\tilde{\xi}_i$ and $\hat{\xi}_i$ are independent $\mathrm{Ber}(p_i)$ random variables corresponding to $\tilde H$ and $\hat H$, respectively.

We would like to deal with the special case $p_i = 1$ separately.
Notice that when $p_i = 1$, $\tilde{\xi}_i - \hat{\xi}_i$ always equals $0$. Therefore, if we set
\[
  \tilde{\zeta}_i := \begin{cases}
      \tilde{\xi}_i & \text{if } p_i < 1;
      \\
      0 & \text{otherwise}
  \end{cases}
  \qquad \text{and} \qquad
  \hat{\zeta}_i :=
  \begin{cases}
      \hat{\xi}_i & \text{if } p_i < 1;
      \\
      0 & \text{otherwise},
  \end{cases}
\]
then $\tilde{\xi}_i - \hat{\xi}_i$ is distributed the same as $\tilde{\zeta}_i - \hat{\zeta}_i$.
Next, notice that $\tilde{\xi}_i - \hat{\xi}_i$ is symmetrically distributed, which in turn implies that
$\tilde{\zeta}_i - \hat{\zeta}_i$
is symmetrically distributed as well. Therefore, $\tilde{\zeta}_i - \hat{\zeta}_i$ is distributed the same as $\eps_i (\tilde{\zeta}_i - \hat{\zeta}_i)$ where $\eps_i$ takes values $+1$, $-1$ with equal probability and is independent of all other random variables.
To summarize,
\[
\E_{\tilde{H}, \hat{H}} \left[ \exp\left(
  \lambda \sup_{ Q_H^{\eta}(z) \le 1 }
    | Q_{\tilde{H}}^{\eta}(z) - Q_{\hat{H}}^{\eta}(z) |
  \right)
  \right]
=
\E_{\tilde{H}, \hat{H}} \E_{(\eps_i)}\left[
  \exp\left(
  \lambda \sup_{ Q_H^{\eta}(z) \le 1 }
    \left| 
    \sum_{i=1}^m \frac{w_i}{p_i} \cdot \eps_i (\tilde{\zeta_i} - \hat{\zeta_i}) \cdot Q_{e_i}(z)
    \right|
  \right)
\right].
\]

Write $G = G_m$. By triangle inequality and rearrangement inequality,
\begin{eqnarray*}
&&
\E_{\tilde{H}, \hat{H}} \E_{(\eps_i)}\left[
  \exp\left(
  \lambda \sup_{ Q_H^{\eta}(z) \le 1 }
    \left| 
    \sum_{i=1}^m \frac{w_i}{p_i} \cdot \eps_i (\tilde{\zeta_i} - \hat{\zeta_i}) \cdot Q_{e_i}(z)
    \right|
  \right)
\right]
\\
&\le&
\E_{\tilde{H}} \E_{(\eps_i)} \left[
  \exp\left(
  2 \lambda \sup_{ Q_H^{\eta}(z) \le 1 }
    \left| 
    \sum_{i=1}^m \frac{w_i}{p_i} \cdot \eps_i \tilde{\zeta_i} \cdot Q_{e_i}(z)
    \right|
  \right)
\right]
\\
&=&
\E_{\tilde{H}} \E_{(\eps_i)} \left[
  \exp\left(
  2 \lambda \sup_{ x \in X}
    \left| 
    \sum_{i=1}^m \frac{w_i}{p_i} \cdot \eps_i \tilde{\zeta_i} \cdot Q_{e_i}((L_G^{\eta})^{-1/2} x)
    \right|
  \right)
\right],
\end{eqnarray*}
where we applied the change of variables $x := (L_G^{\eta})^{1/2} z$ and $X := \{x \in \R^V: Q_H^{\eta}((L_G^{\eta})^{-1/2} x) \le 1\}$.
Note that by \autoref{prop:energy-comparison}, $X \subseteq B_2^n$ where $B_2^n$ is the closed unit ball in $\R^n$.

\subsection{Setting up the Metric}

Consider the inner expectation
\[
\E_{(\eps_i)} \left[
  \exp\left(
  2 \lambda \sup_{ x \in X}
    \left| 
    \sum_{i=1}^m \frac{w_i}{p_i} \cdot \eps_i \tilde{\zeta_i} \cdot Q_{e_i}((L_G^{\eta})^{-1/2} x)
    \right|
  \right)
\right].
\]
For each fixed $\tilde{H}$ (which means fixing the $\tilde{\zeta}_i$'s), let
\[
  V_x := 
    \sum_{i=1}^m \frac{w_i}{p_i} \cdot \eps_i \tilde{\zeta_i} \cdot Q_{e_i}((L_G^{\eta})^{-1/2} x).
\]

We would like to apply \autoref{prop:azuma} to control the difference between $V_x$ and $V_y$ for any two points $x, y \in X$.
Set
\[
  \psi_i := \sum_{j=1}^i \frac{w_j}{p_j} \cdot \eps_j \tilde{\zeta}_j \cdot (Q_{e_j} ((L_G^{\eta})^{-1/2}x) - Q_{e_j} ((L_G^{\eta})^{-1/2} y)).
\]
Note that $\psi_0 = 0$ and $\psi_m = V_x - V_y$, and that $(\psi_i)$ is a martingale. The difference $|\psi_i - \psi_{i-1}|$ is always bounded by $M_i := (w_i / p_i) \cdot \tilde{\zeta}_i \cdot (Q_{e_i} ((L_G^{\eta})^{-1/2}x) - Q_{e_i} ((L_G^{\eta})^{-1/2} y))$.
Then, \autoref{prop:azuma} gives
\[
  \Pr \left[
    |V_x - V_y| \ge \beta
  \right]
  \le
  2 \exp \left(
    \frac{-\beta^2}{2 d(x, y)^2}
  \right),
\]
where
\[
d(x, y) := \sqrt{
  \sum_{i=1}^m \left( \frac{w_i}{p_i}\right)^2 \cdot \tilde{\zeta_i}^2 \cdot (Q_{e_i} ((L_G^{\eta})^{-1/2}x) - Q_{e_i} ((L_G^{\eta})^{-1/2} y))^2
}.
\]
Since $d(x, y)$ is the $\ell_2$ distance between the images of $x$ and $y$ under a mapping, it is a semi-metric on $X$.

\subsection{Bounding Chaining Functional}

In order to apply \autoref{lem:chaining-whp}, we would like to upper bound both $\gamma_2(X, d)$ and $\Delta(X, d)$.
We first bound $\gamma_2(X, d)$.
The following chaining guarantee is the key bound in~\cite{Lee2023}.

\begin{proposition}[Bound on $\gamma_2(X, d)$; see~{\cite[Corollary~2.13]{Lee2023}}]\label{prop:gamma-bound}
  Suppose $X \subseteq B_2^n$ and that a metric $d(\cdot, \cdot)$ of the form
  \[
    d(x, x') := \left(
      \sum_{i=1}^m \left|
        \phi_i(Ax)^2 - \phi_i(Ax')^2
      \right|^2
    \right)^{1/2},
  \]
  where $A: \R^n \rightarrow \R^M$ is a linear map and $\phi_1, \dots, \phi_m: \R^M \rightarrow \R$ are semi-norms in the form of 
  \[ 
    \phi_i(z) = \max_{j \in S_i} \omega_j \abs{(Az)_j}
  \]
  for some $S_i \subseteq [M]$ and $\omega_j \geq 0$ ($j \in S_i$).
  Let $\alpha > 0$ be a constant such that $|\phi_i(z) - \phi_i(z')| \le \alpha \norm{z - z'}_{\infty}$ for all $i \in [m]$. 
  Let $r = \max_{i \in [m]} \abs{S_i}$.
  Then,
  \[
  \gamma_2(X, d) \lesssim 
  \alpha \sqrt{\log (M + n) \cdot \log r} \norm{A}_{2 \rightarrow \infty}
  \cdot
  \sup_{x \in X} \left(
    \sum_{i=1}^m \phi_i(Ax)^2
  \right)^{1/2}.
  \]
\end{proposition}

In order to apply \autoref{prop:gamma-bound}, we define $A: \R^V \rightarrow \R^{V \times V}$ as
\[
  (Ax)_{uv} :=
  \left\langle x, \frac{(L_G^{\eta})^{-1/2} (\chi_u - \chi_v)}{\norm{(L_G^{\eta})^{-1/2} (\chi_u - \chi_v)}_2} \right\rangle
\]
and
\[
  \phi_i(z) := \sqrt{\frac{w_i \tilde{\zeta}_i}{p_i}} 
  \max_{u, v \in e_i} \left( \norm{(L_{G}^{\eta})^{-1/2}(\chi_u - \chi_v)}_2 \cdot|z_{uv}| \right).
\]
Then,
\[
  \phi_i(Ax)^2 = \frac{w_i \tilde{\zeta}_i}{p_i} \cdot
  \max_{u, v \in e_i} \left\langle x, (L_G^{\eta})^{-1/2} (\chi_u - \chi_v) \right\rangle^2
  =
  \frac{w_i \tilde{\zeta}_i}{p_i} \cdot
  Q_{e_i}((L_G^{\eta})^{-1/2} x),
\]
so our metric does take the form in the proposition.
Using the fact that $p_i = c r_i w_i$ when $\tilde{\zeta}_i \neq 0$, and that
\[
  r_i := \max_{u', v' \in e_i} \norm{(L_{G_i}^{\eta})^{-1/2}(\chi_{u'} - \chi_{v'})}_2^2
  \ge \norm{(L_{G}^{\eta})^{-1/2}(\chi_u - \chi_v)}_2^2,
\]
we can see that $|\phi_i(z) - \phi_i(z')| \le 1/\sqrt{c} \cdot \norm{z - z'}_{\infty}$. Next, since each row of $Ax$ is formed by taking the inner product of $x$ with a unit vector, we have that $\norm{A}_{2 \rightarrow \infty} \le 1$.
Since $M = |V \times V| = n^2$, \autoref{prop:gamma-bound} then gives
\[
  \gamma_2(X, d) \lesssim \sqrt{\frac{\log n \log r}{c}} \cdot \sup_{x \in X} \left(
    \sum_{i=1}^m \frac{w_i \tilde{\zeta}_i}{p_i} Q_{e_i}((L_{G}^{\eta})^{-1/2} x)
  \right)^{1/2}
  \leq
  \sqrt{\frac{\log n \log r}{c}} \cdot \sup_{z: Q_H(z) \le 1} Q_{\tilde{H}}(z)^{1/2}.
\]

\subsection{Bounding Diameter}
Next, we upper bound the diameter $\Delta(X, d)$, which amounts to upper bounding $2 \sup_{x \in X} d(x, \vec{0})$.
Indeed,
\begin{eqnarray*}
    \sup_{x \in X} d(x, \vec{0})
    &=&
    \sup_{x \in X}
    \sqrt{
     \sum_{i=1}^m \left( \frac{w_i}{p_i}\right)^2 \cdot \tilde{\zeta_i}^2 \cdot (Q_{e_i} ((L_G^{\eta})^{-1/2}x))^2
    }
    \\
    &\le&
    \sup_{x \in X} \sqrt{
    \left(
     \sup_{i \in [m]} \frac{w_i \tilde{\zeta}_i}{p_i} \cdot Q_{e_i} ((L_G^{\eta})^{-1/2}x)
    \right)
     \cdot
    \left(
     \sum_{i=1}^m \frac{w_i \tilde{\zeta}_i}{p_i} Q_{e_i} ((L_G^{\eta})^{-1/2}x)
    \right)
    }
    \\
    &\le&
    \sqrt{
      \left( \sup_{x \in X} \sup_{i \in [m]} \frac{w_i \tilde{\zeta}_i}{p_i} \cdot
      \sup_{u, v \in e_i} \langle x, (L_G^{\eta})^{-1/2} (\chi_u - \chi_v) \rangle^2
      \right) 
      \cdot
      \sup_{z: Q_H(z) \le 1} Q_{\tilde{H}}(z)
    }
    \\
    &\le&
    \frac{1}{\sqrt{c}} \cdot \sup_{z: Q_H(z) \le 1} Q_{\tilde{H}}(z)^{1/2},
\end{eqnarray*}
where we used the fact that $p_i = c w_i r_i$ if $\tilde{\zeta}_i \neq 0$ and
\[
  \langle x, (L_G^{\eta})^{-1/2} (\chi_u - \chi_v)\rangle^2
  \le
  \norm{x}_2^2 \cdot \norm{(L_G^{\eta})^{-1/2} (\chi_u - \chi_v)}_2^2
  \le
  \norm{(L_G^{\eta})^{-1/2} (\chi_u - \chi_v)}_2^2
  \le
  r_i.
\]
The second inequality is because $X \subseteq B_2^n$ and the third inequality is because $r_i \ge \norm{(L_{G_i}^{\eta})^{-1/2}(\chi_u - \chi_v)}_2^2 \ge \norm{(L_{G}^{\eta})^{-1/2}(\chi_u - \chi_v)}_2^2$ for all $i \in [m]$ and $u, v \in e_i$.

We thus conclude that $\Delta(X, d) \le O(1/\sqrt{c})$.

\subsection{Conclusion}
For each fixed $\tilde{H}$, apply \autoref{lem:chaining-whp} with the previously obtained upper bounds on $\gamma_2(X, d)$ and $\Delta(X, d)$, noting that
\[
\sup_{z: Q_H(z) \le 1} Q_{\tilde{H}}(z)^{1/2}
\le
\left( 1 + \sup_{z: Q_H(z) \le 1} |Q_H(z) - Q_{\tilde{H}}(z)| \right)^{1/2}
= (1 + Z)^{1/2}.
\]
Now, apply the outer expectation over $\tilde{H}$, and \autoref{lem:exponential-mgf-bound} follows.
\end{document}